\documentclass[runningheads]{llncs}

\usepackage[T1]{fontenc}
\usepackage[utf8]{inputenc}
\usepackage{amsmath,amssymb}

\usepackage{graphicx}
\usepackage{epstopdf}
\DeclareGraphicsExtensions{.pdf,.png,.jpg,.jpeg,.eps}
\usepackage{wrapfig}
\usepackage{picinpar}
\usepackage{multirow}
\usepackage{tabularx}
\usepackage{colortbl}
\usepackage{tikz}

\usepackage[ruled,linesnumbered,vlined]{algorithm2e}

\usepackage{xspace}
\usepackage{float}
\usepackage{url}

\newtheorem{myTheo}{Theorem}

\begin{document}
\title{Dual-Metric Partitioning with Adaptive Kernel Execution
	for Efficient GCN Acceleration}

\titlerunning{Dual-Metric Partitioning for Efficient GCN Acceleration}

\author{
	Lingling Zhang\inst{1} \and
	Hang Zeng\inst{1} \and
	Pengpeng Qiao\inst{2}\thanks{Corresponding author.} \and
	Zhiwei Zhang\inst{3} \and
	Ye Yuan\inst{3} \and
	Guoren Wang\inst{3}
}

\authorrunning{L. Zhang et al.}

\institute{
	Capital Normal University, Beijing, China\\
	\email{\{7089,2251002056\}@cnu.edu.cn}
	\and
	Institute of Science Tokyo, Tokyo, Japan\\
	\email{peng2qiao@gmail.com}
	\and
	Beijing Institute of Technology, Beijing, China\\
	\email{\{zwzhang,yuan-ye\}@bit.edu.cn, wanggrbit@126.com}
}

%
%
\maketitle              
\begin{abstract}
Graph Convolutional Networks (GCNs) are widely used for large
graph-structured data, including social, citation, and e-commerce
networks, but their deployment is constrained by irregular memory
access and severe GPU workload imbalance. These challenges arise in
two dimensions: width imbalance from power-law degree distributions
and depth imbalance from heterogeneous neighborhood connectivity. We
present DualGCN, a GPU acceleration framework addressing both
dimensions through dual-metric graph partitioning and adaptive kernel
execution. DualGCN combines node degree, reflecting aggregation width,
with neighborhood density estimated by anonymous random walks,
capturing multi-hop connectivity and access depth. This hybrid workload
metric enables connectivity-aware partitioning of large graphs into
sparse and dense regions while reducing workload imbalance from linear
to logarithmic complexity. DualGCN then selects partition-specific
execution strategies: sparse partitions use warp-level parallelism and
coalesced memory access, whereas dense partitions exploit
instruction-level parallelism to hide latency and improve GPU
utilization. Experiments on twelve real-world graph datasets show that
DualGCN consistently accelerates GCN computation, achieving average
speedups of $2.53\times$, $3.8\times$, and $2.13\times$ over cuSPARSE,
GNNAdvisor, and ACCEL, respectively. These results demonstrate that
jointly optimizing graph partitioning and kernel execution provides an
effective solution for processing large-scale graph and social-network
workloads.
\keywords{Graph convolutional networks \and GPU acceleration \and graph partitioning \and sparse--dense matrix multiplication}
\end{abstract}
\section{Introduction}
\label{sec:introduction}
Graph Convolution Networks (GCNs) represent an extension of deep learning techniques to the domain of graph-structured data, enabling the processing and analysis of relational information \cite{wu2020comprehensive,qiao2024astore}. GCNs excel in applications such as node classification, recommendation systems, and traffic forecasting, demonstrating their versatility and effectiveness in processing interconnected data. The deployment of GCNs in these applications imposes strict constraints on latency and throughput \cite{wang2023tc,gurevin2024prunegnn}. Consequently, GPU platforms have become the predominant choice for designing and accelerating GCN training, owing to their ability to meet these performance demands. Building on this trend, GCN accelerators on GPUs have attracted significant attention by focusing on GPU designs that are adaptable to graph processing.

Conventional GCN accelerators execute two alternating phases
\cite{wang2023tc}: an \emph{aggregation phase}, in which a sparse and
irregular adjacency matrix $\mathbf{A}$ is multiplied by a dense node
feature matrix $\mathbf{X}$ through sparse--dense matrix multiplication
(SpMM), and an \emph{update phase}, which performs neural-network
operations, primarily multiplication with a small dense weight matrix
$\mathbf{W}$. Because SpMM accounts for more than 80\% of GCN training
time \cite{wang2023tc,wang2021gnnadvisor}, most GCN accelerators focus
on optimizing this operation. Existing approaches improve GCN
computation through algorithmic techniques
\cite{wang2021gnnadvisor,peng2024maxk}, architecture-level
optimizations \cite{chen2022regnn}, or combinations of both
\cite{gurevin2024prunegnn,you2022gcod,xie2023accel}. Algorithmic
approaches reorganize or partition graph data using techniques such as
runtime community detection and graph partitioning to improve workload
balance and data locality under power-law degree distributions
\cite{wang2021gnnadvisor,xie2023accel}. Architecture-level approaches
manage thread blocks, warps, and shared memory to mitigate workload
imbalance caused by irregular graph structures
\cite{you2022gcod,peng2024maxk}.

State-of-the-art GCN accelerators face algorithmic and architectural
challenges that limit scalability and efficiency. Graph-partitioning
frameworks based on clustering coefficients or community detection
incur prohibitive time complexity of at least
$O(\frac{m^2}{n})$~\cite{hwang2023grow,ccatalyurek2012multithreaded},
where $m$ and $n$ denote edge and vertex counts, respectively. To
reduce this complexity, existing frameworks rely solely on node degree
information from the adjacency matrix's power-law distribution
\cite{wang2021gnnadvisor,xie2023accel}, overlooking neighbor interaction
patterns and multi-hop aggregation dependencies that affect
computational performance and efficiency. Similarly, architecture-level
frameworks employ simplified GPU resource management using static warp
assignments based only on adjacency-matrix row sparsity
\cite{chen2022regnn}, failing to adapt to runtime workload variations
or jointly optimize feature-matrix dimensions and nonzero-element
distributions for dynamic workload balancing and resource utilization
\cite{xie2023accel,peng2024maxk}.

To address these challenges, we propose DualGCN, a hybrid
algorithmic--architectural design that coordinates graph partitioning
with adaptive GPU execution. Algorithmically, DualGCN analyzes graph
connectivity patterns to partition workloads into sparse and dense
regions. Architecturally, it employs an adaptive dual-path strategy
that assigns blocks and warps according to feature dimensionality and
nonzero-element distribution. This integrated design improves data
locality, balances workloads, and increases overall GPU utilization
across diverse graph structures. Our system makes the following
contributions:

\textbullet\  \textbf{Connectivity-aware Data Partitioning:} We propose a data partitioning scheme based on anonymous random walks to efficiently detect and analyze neighborhood connectivity patterns. By jointly considering node connectivity and degree distributions, our method can divide the graph data into sparse and dense parts and achieve balanced partitioning with linear time complexity, significantly improving data locality during GCN computations.

\textbullet\ \textbf{Dynamic Resource Management:} We develop a dual-path workload distribution strategy that adaptively allocates blocks and warps based on workload characteristics. This DualGCN strategy achieves superior load balance and memory access efficiency across diverse graph structures, developed into a kernel algorithm. The kernel manages GPU resources according to the workload distribution associated with the data partitions, efficiently exploiting parallelism when executing the SpMM computations in GCN models.

\textbullet\  \textbf{Comprehensive Evaluation:} Through extensive experiments on real-world graph datasets, we demonstrate that DualGCN outperforms the state-of-the-art methods, such as cuSPARSE \cite{naumov2010cusparse}, GNNAdvisor \cite{wang2021gnnadvisor}, and ACCEL \cite{xie2023accel} by factors of $2.53\times$, $3.8\times$, and $2.13\times$, on average. The results validate the effectiveness of our joint algorithm-architecture optimization approach in accelerating GCN computations while maintaining balanced resource utilization.
\section{Background and Motivation}
\subsection{Graph Convolution Networks}\label{sec2_gcn}
Graph Convolution Networks (GCNs) process graph-structured data through iterative neighborhood aggregation and feature transformation. For graph $G=(V,E)$ with $|V|$ nodes and $|E|$ edges, where each node has an $F$-dimensional feature vector, GCNs compute node embeddings to capture both local structure and node features. As shown in Figure \ref{GNN_flow}, each GCN layer transforms node $v$'s embedding at layer $k+1$ ($h_v^{(k+1)}$) by combining: (i) the node's features from layer k; (ii) aggregated features from neighboring nodes $N(v)$; and (iii) edge features $e_{vu}$ connecting node $v$ to neighbor $u$. The computation is: $h_v^{(k+1)} = f\!\left(h_v^{(k)},\{(h_u^{(k)},e_{vu})\mid u\in N(v)\}\right)$, where $f$ is a learnable update function, typically a neural network. Through multiple layers (Figure \ref{GNN_flow}), GCNs progressively capture wider neighborhood information, learning representations that reflect both local connectivity and global graph properties.
\begin{figure}[hbt!]
	\centering
	\includegraphics[height=0.38\textwidth]{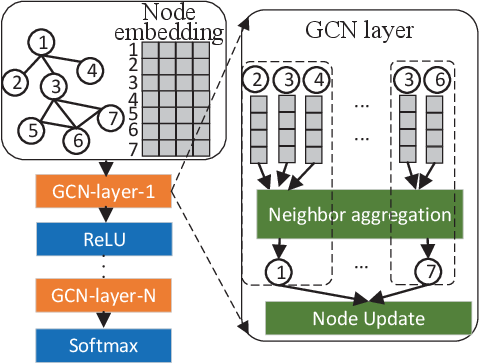}
	\caption{An Example of GCN General Computation Flow.}\label{GNN_flow}
	\vspace{-2.5em}
\end{figure}

GCN variants like GraphSAGE \cite{hamilton2017inductive} and Graph Isomorphism Network (GIN) \cite{xu2018powerful} employ different aggregation functions but follow the same forward propagation model as standard GCNs. The feature aggregation phase, implemented as Sparse-Dense Matrix-Matrix Multiplication (SpMM), dominates GCN computation time---often consuming over 80\% of execution cycles due to its highly irregular memory access patterns. This computational bottleneck severely limits GCN scalability and real-time deployment. Since SpMM is fundamental to all GCN architectures and represents the primary performance barrier, accelerating its execution is essential for practical GCN applications.
\subsection{Related Work}
\textbf{Graph Partitioning Approaches.} Recent GCN acceleration frameworks have explored graph partitioning to improve data locality and workload balance. GNNAdvisor \cite{wang2021gnnadvisor} pioneered community-based partitioning combined with weight matrix dimensions for warp-level workload distribution, though it underperforms NVIDIA's cuSPARSE \cite{naumov2010cusparse}. GROW \cite{hwang2023grow} introduced row-wise product acceleration with clustered caching for high-degree nodes, but its $O(|V|^2)$ partitioning complexity limits scalability. Moreover, their static partitioning criteria cannot capture heterogeneous neighborhood structures, causing suboptimal workload assignments across graph topologies and feature dimensions under rapidly changing real-world workloads. These methods demonstrate that while partitioning improves locality, computational overhead and imbalanced partitions remain significant challenges for large-scale graphs.

\textbf{GPU Resource Management.} Several frameworks optimize GPU resource allocation to address workload imbalance. MergePath \cite{shan2023mergepath} evenly distributes non-zero elements across threads but requires hardware-specific parameter tuning. GCoD \cite{you2022gcod}, ACCEL \cite{xie2023accel}, and MaxK-GNN \cite{peng2024maxk} employ degree-based partitioning with optimized block and warp management, achieving better load balancing than static approaches. However, these methods still struggle with power-law degree distributions, where a few high-degree nodes create severe workload skew despite sophisticated scheduling strategies.

\textbf{Memory Optimization.} MEGA \cite{zhu2024mega} targets memory efficiency through degree-aware quantization, adaptive-package storage formats, and Condense-Edge scheduling for optimized access patterns. While these techniques reduce memory bandwidth requirements and improve cache utilization, they introduce training overhead and require complex mixed-precision hardware support. This trade-off between memory efficiency and computational complexity highlights the need for solutions that jointly optimize memory access and computational workload without excessive overhead.
\subsection{Motivation}

The limitations of existing GCN accelerators reveal three interdependent performance bottlenecks: data partitioning, workload allocation, and memory access patterns.

\textbf{Data Partitioning Challenges.} Graph adjacency matrices exhibit power-law degree distributions and heterogeneous connectivity patterns \cite{fu2022tlpgnn,xu2025continuous}. Traditional schemes consider only node degrees, overlooking connectivity. Figure \ref{GCN_example} illustrates---nodes '1' and 'a' have identical degrees but different neighborhood structures. In a two-layer GCN, node '1' aggregates from a dense neighborhood while node 'a' processes a star topology, creating imbalance despite equal degrees. \emph{Effective partitioning must capture both degree distribution and connectivity patterns.}
\begin{figure}[hbt!]
	\centering
	\includegraphics[height=0.38\textwidth]{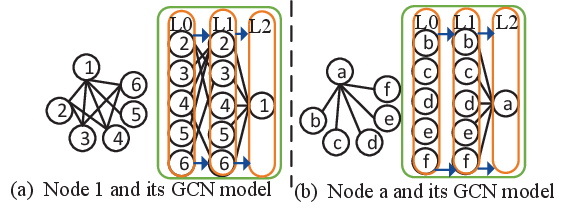}
	\caption{Two-layer GCN computation: (a) Node 1 with dense connectivity requires complex aggregation, (b) Node a with star topology requires simpler computation.}\label{GCN_example}
	\vspace{-2.5em}
\end{figure}

\textbf{Workload Allocation Inefficiencies.} GPUs offer multiple parallelism levels---threads, warps, and blocks---with data parallelism across partitions. Current frameworks optimize these independently, missing coordination opportunities. \emph{The challenge is mapping irregular graph workloads to regular GPU execution patterns while maintaining load balance.}

\textbf{Memory Access Inefficiencies.} Existing methods apply uniform memory strategies regardless of partition characteristics \cite{xie2023accel,peng2024maxk}. Sparse partitions suffer random access while dense partitions encounter bandwidth saturation. \emph{Adaptive strategies tailored to partition density could significantly improve throughput.}

These factors form a coupled optimization space---partitioning affects workload distribution, determining optimal memory patterns. \emph{High-performance GCN acceleration requires joint optimization across all dimensions, motivating our integrated approach that coordinates partitioning, workload allocation, and memory access based on graph characteristics.}
\section{DualGCN Framework}
\label{GNNA}
\subsection{Overview of DualGCN}
DualGCN (Graph Convolution Network Accelerator) is a framework built on two core components: data partitioning and workload distribution. The data partitioning component strategically divides graph data across GPU blocks to achieve balanced computational loads, segmenting the graph into dense and sparse regions based on node connectivity. The workload distribution component then optimizes how this partitioned data is processed, implementing specialized resource management strategies for SpMM computations. In sparse regions, each GPU block processes multiple rows using its assigned warps, while in dense regions, each block handles a single row that is further subdivided and assigned to warps - all coordinated with efficient memory management. 
Figure \ref{GNN_arch} illustrates an overview of the DualGCN architecture for optimized SpMM operations. It demonstrates how the multiplication between a sparse matrix (Figure \ref{GNN_arch}(a)) and a dense matrix (Figure \ref{GNN_arch}(b)) is executed through a dual-path workload assignment strategy (Figure \ref{GNN_arch}(c)). In this scheme, computational blocks are dynamically allocated: each block processes either multiple sparse rows or a single dense row, depending on the input matrix's density characteristics.
\begin{figure}[hbt!]
	\centering
	\includegraphics[height=0.38\textwidth]{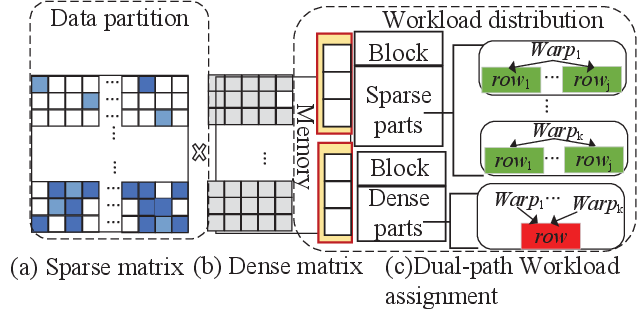}
	\caption{An Example of DualGCN Architecture.}\label{GNN_arch}
	\vspace{-2.5em}
\end{figure}
\subsection{Adaptive Data Partitioning}

\textbf{Motivation and Key Contribution:} Existing graph partitioning schemes for GCN acceleration typically rely on simple heuristics (e.g., node degree alone) that fail to capture the complex computational patterns in neighborhood aggregation. We present a theoretically-grounded adaptive partitioning scheme that jointly optimizes for both structural complexity and computational load balancing, enabling efficient dynamic resource allocation with provable performance bounds.

\subsubsection{Theoretical Foundation: Anonymous Random Walk Analysis}

We employ anonymous random walks to quantify the computational complexity of neighborhood aggregation beyond simple degree metrics. An anonymous random walk iteratively traverses a graph from a node to random neighbors, labeling each node by its first occurrence order in the sequence.

\begin{myTheo}\label{ARWalks}
Let $S = (\nu, d)$ be the set of anonymous random walks after $d$ steps, starting from node $\nu$. Let $P_d$ be the distribution of the number of unique nodes on the paths of these random walks. The sum $\sum_{i=1}^d iP_i$ characterizes the structural density of the subgraph induced by all nodes $\mu$ such that $dist(\nu,\mu) \leq d$, and provides a tighter bound on aggregation complexity than degree-based metrics alone.
\end{myTheo}

The proof follows from \cite{micali2016reconstructing}. \textbf{Key insight:} Theorem \ref{ARWalks} reveals that anonymous random walks capture multi-hop connectivity patterns critical for GCN computation that simple degree metrics miss, enabling more accurate workload prediction.

\subsubsection{Dual-Metric Workload Characterization}

We introduce a novel dual-metric approach that combines complementary graph properties:

\textbf{\emph{1. Neighbor Density ($\mathcal{N}_\nu$):}} We define the neighbor density of node $\nu$ as $\mathcal{N}_\nu = \sum_{i=1}^d iP_i$ from Theorem \ref{ARWalks}. This metric captures the \emph{computational depth} of neighborhood aggregation, accounting for indirect connections that impact memory access patterns and cache efficiency. Unlike degree-only metrics, neighbor density predicts the actual number of unique feature vectors accessed during multi-layer GCN propagation.

\textbf{\emph{2. Node Degree ($\mathcal{D}_\nu$):}} While graphs typically follow power-law degree distributions \cite{sala2010brief}, degree alone provides an incomplete picture. We use degree to capture \emph{computational width}---the immediate parallelization opportunity for neighbor aggregation.

\subsubsection{The Dual-Metric Approach}
\paragraph{Approach Overview.}
We propose a dual-metric approach combining complementary measures to capture GCN workload characteristics. Our hybrid metric $\mathcal{H}_\nu = \alpha\mathcal{N}_\nu + (1-\alpha)\mathcal{D}_\nu$ integrates node degree $\mathcal{D}_\nu$ (width) with neighborhood density $\mathcal{N}_\nu$ (depth), where $\alpha$ is a weighting parameter. The degree metric $\mathcal{D}_\nu$ corresponds to neighbors requiring aggregation, while density $\mathcal{N}_\nu$ is computed through $k$ anonymous random walks of length $d$ to capture local structure beyond immediate neighbors. This enables accurate workload prediction by considering both operation quantity (degree) and memory access complexity (density).
\paragraph{Algorithmic Strengths and Performance Guarantees.}
Our dual-metric approach addresses both workload imbalance dimensions simultaneously. Width imbalance from power-law degree distributions causes high-degree nodes to require more parallel threads. Depth imbalance stems from varying neighborhood densities creating disparate memory access patterns---dense neighborhoods generate cache misses and bandwidth contention. Degree-only approaches miss this depth dimension, resulting in suboptimal distribution among nodes with similar degrees but different neighborhood characteristics. Our partitioning minimizes $\sum_{i=1}^{B} \text{Var}(\{\mathcal{H}_\nu : \nu \in \text{Block}_i\})$ across $B$ blocks, achieving theoretical bounds: workload imbalance reduces from $O(\Delta)$ in naive partitioning to $O(\log\Delta)$, where $\Delta$ is maximum degree and imbalance is $\frac{\max_i W_i}{\text{avg}_i W_i}$ for block workloads $W_i$ (Theorem~\ref{LoadBalance}).
\begin{myTheo}\label{LoadBalance}
For a graph with maximum degree $\Delta$ and our hybrid partitioning, the workload imbalance factor is bounded by $O(\log\Delta)$ compared to $O(\Delta)$ for naive partitioning, where workload imbalance is defined as $\frac{\max_i W_i}{\text{avg}_i W_i}$ for block workloads $W_i$.
\end{myTheo}

\textit{Proof Sketch.}
The proof establishes that naive (random) partitioning achieves $O(\Delta)$ imbalance because high-degree nodes with workload $\Theta(\Delta)$ can randomly cluster in certain blocks---with some blocks receiving multiple high-degree nodes while others receive none, creating imbalance proportional to the maximum degree $\Delta$. In contrast, our hybrid metric approach first sorts all nodes by their combined degree-density score $\mathcal{H}_\nu$, creating a smoothly increasing sequence where adjacent nodes have similar workloads. After sorting nodes by $\mathcal{H}_\nu$, we have the sequence $\mathcal{H}_1 \leq \mathcal{H}_2 \leq ... \leq \mathcal{H}_n$. The key insight is that the total workload mass $\sum_{i=1}^n \mathcal{H}_i = O(n \cdot \text{average degree}) = O(n)$ in sparse graphs. For power-law graphs, only $O(n/\Delta)$ nodes have degree $\geq \Delta/2$, so the workload distribution is not uniform. The steepest increase occurs where high-degree nodes appear in the sorted order. Even in the worst case where all high-degree nodes appear consecutively, the total increase over a range of $\ell = n/\log n$ positions is at most $O(\Delta)$. Therefore, the average gap per position is $\frac{O(\Delta)}{\ell} = \frac{O(\Delta)}{n/\log n} = O(\frac{\Delta \log n}{n})$. When we partition this sorted sequence into $B$ contiguous blocks, each block contains nodes with similar hybrid metric values, ensuring the maximum difference between blocks is bounded by $O(n/B) \times O(\Delta \log n/n) = O(\Delta \log n/B)$. Since the average block workload is $\Omega(n/B)$ and we choose $B = \Omega(\Delta)$ blocks for GPU execution, the imbalance factor becomes $\text{max/avg} = O(\log n) = O(\log \Delta)$, achieving logarithmic rather than linear scaling with the maximum degree.

\paragraph{Time Complexity and Overhead Analysis.}
The preprocessing phase consists of four steps with well-defined complexity bounds. Parallel density computation requires $O(kd)$ time with $O(|V|)$ processors for computing $\mathcal{N}_\nu$ across all vertices. Degree extraction from CSR row pointers takes $O(1)$ per vertex. The adaptive integration phase, which dynamically adjusts $\alpha$ based on variance analysis (defaulting to 0.5), requires $O(|V|)$ time. Finally, sorting and partitioning nodes requires $O(|V|\log|V|)$ time. The total preprocessing complexity of $O(|V|(kd + \log|V|))$ with constants $k,d = O(1)$ is amortized over $T$ training epochs, yielding only $O(\frac{kd + \log|V|}{T})$ overhead per iteration.
\subsection{DualGCN Kernel Design and Implementation}

\subsubsection{Kernel Architecture Overview}
Our kernel leverages dual-metric partitioning for execution. The hybrid metric $\mathcal{H}_\nu$ classifies nodes by degree (width) and neighborhood density (depth). High $\mathcal{H}_\nu$ blocks use dense processing with 4-way ILP; low $\mathcal{H}_\nu$ blocks use sparse processing.

The architecture employs three parallelism levels. Block-level distribution follows variance-minimized partitioning for balanced SM computation. Warp-level processing assigns resources by $\mathcal{D}_\nu$ values---high-degree nodes receive multiple warps ($\lceil \text{row\_nnz}/32 \rceil$), low-degree nodes share warps. Thread-level computation uses $\mathcal{N}_\nu$ values, with dense neighborhoods utilizing ILP.

Three principles guide execution. First, adaptive execution selects paths by partition's $\mathcal{H}_\nu$. Second, memory optimization matches partition characteristics---sparse uses L2 cache via \texttt{\_\_ldg()}, dense allocates 512B-2KB shared memory. Third, hardware utilization follows metrics: $\mathcal{D}_\nu$ determines warps, $\mathcal{N}_\nu$ influences registers and ILP. This creates metric-aware resource allocation throughout execution.
\subsubsection{Detailed Kernel Components}
Our adaptive kernel architecture employs a dual-path strategy that
dynamically selects sparse or dense execution modes based on workload
characteristics. The mechanism uses the hybrid metric $\mathcal{H}$
computed during data partitioning to optimize resource utilization for
each thread block.

\paragraph{Processing Mode Selection.}
Upon kernel invocation, each thread block determines its execution path through the following decision criterion:
\begin{equation}
\text{Processing Mode} = \begin{cases}
\text{Sparse Path} & \text{if } \mathcal{H} < \tau \\
\text{Dense Path} & \text{if } \mathcal{H} \geq \tau
\end{cases}
\end{equation}
where the threshold $\tau$ (ProBound) is empirically calibrated as $\tau = \mu_{\mathcal{H}} + 0.5\sigma_{\mathcal{H}}$, representing the mean plus half the standard deviation of the workload distribution. This threshold effectively separates irregular, sparse computations from regular, dense operations, enabling specialized optimizations for each pattern.

\paragraph{Sparse Block Processing.}
Sparse partitions exhibit irregular memory access patterns and load imbalance, necessitating a specialized processing strategy. We address these challenges through a warp-feature association mechanism that prioritizes memory coalescing while maintaining computational efficiency.

In this approach, we establish a three-tier hierarchy: (1) each warp is assigned a subset of matrix rows, (2) warps process their assigned rows for one feature at a time, and (3) all warps within a block synchronously advance through the same feature range. This synchronization ensures that threads within a warp access contiguous memory regions, achieving coalesced memory transactions. Features are processed in chunks of 32 elements, aligning with the warp size to eliminate divergent execution paths.

Figure~\ref{GCN_parallel}(a) illustrates this workload distribution strategy. The left panel shows the initial assignment of data elements to three warps ($\text{warp}_1$, $\text{warp}_2$, $\text{warp}_3$), with each warp managing three rows across three features ($F_1$, $F_2$, $F_3$). The right panel demonstrates the feature-level parallelization phase, where rows undergo weighted aggregation with feature-specific weights ($w_f$), color-coded to distinguish different computational operations.

\begin{figure}[hbt!]
	\centering
	\includegraphics[height=0.25\textwidth]{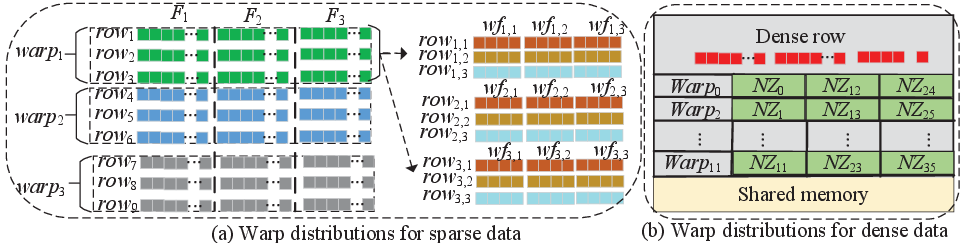}
	\caption{Workload distribution in sparse block processing: (a) warp-to-row assignment and (b) feature-level parallel execution.}\label{GCN_parallel}
	\vspace{-2.5em}
\end{figure}

\paragraph{Dense Block Processing with 4-Way ILP.}
Dense partitions, characterized by regular access patterns and balanced workloads, benefit from aggressive parallelization techniques. Our dense processing path exploits this regularity through a combination of collaborative warp execution and instruction-level parallelism (ILP).

The key innovation lies in our 4-way ILP implementation, where each thread simultaneously processes four independent data elements. This design leverages the GPU's instruction pipeline to overlap memory operations with arithmetic computations, effectively hiding memory latency. Modern GPU architectures support multiple outstanding memory requests per thread, enabling concurrent data fetching while previous computations complete. This organization increases arithmetic intensity, reduces pipeline stalls, and improves occupancy without introducing synchronization overhead among threads within the warp.

Algorithm~\ref{SpMMDense4Way} details our implementation strategy.
Each thread maintains four independent accumulators
($\text{sums}[0..3]$) to eliminate data dependencies, allowing the GPU
scheduler maximum flexibility in instruction reordering. The strided
access pattern, with
$\mathtt{stride}=\mathtt{block\_size}\times 4$, ensures that consecutive
threads access consecutive memory addresses, achieving perfect memory
coalescing. Additionally, all memory operations utilize
\texttt{\_\_ldg()} intrinsics to bypass the L1 cache and directly access
the read-only data cache, optimizing bandwidth utilization for these
predictable access patterns.

\begin{algorithm}[h]
\caption{Dense Block Processing with 4-Way ILP}
\label{SpMMDense4Way}
\SetAlgoLined
\KwIn{Thread ID \texttt{tid}, block size \texttt{block\_size}, sparse values array, column indices, dense matrix $\mathbf{B}$, feature index}
\KwOut{Accumulated sums for four elements}
    \texttt{stride} $\gets$ \texttt{block\_size} $\times$ 4\;
    \texttt{sums}[0..3] $\gets$ 0 \tcp*{Initialize four accumulators}
    \For{$k \gets 0$ \KwTo $3$}{
        \texttt{idx} $\gets$ \texttt{tid} + $k \times$ \texttt{block\_size}\;
        \If{\texttt{idx} $<$ \texttt{nnz}}{
            \texttt{col} $\gets$ \texttt{\_\_ldg}(\texttt{col\_idx}[\texttt{idx}])\;
            \texttt{val} $\gets$ \texttt{\_\_ldg}(\texttt{values}[\texttt{idx}]) $\times$ \texttt{\_\_ldg}($\mathbf{B}$[\texttt{col}][\texttt{feat}])\;
            \texttt{sums}[$k$] $\gets$ \texttt{sums}[$k$] + \texttt{val}\;
        }
    }
    \KwRet{\texttt{sums}}\;
\end{algorithm}

This dual-path architecture ensures that both sparse and dense data patterns receive optimized treatment, maximizing overall kernel performance across diverse workload distributions.

\subsubsection{Memory Management Strategy}
\paragraph{Overview.}
Our memory management strategy addresses the fundamental challenge of maximizing memory bandwidth utilization while minimizing access latency in GPU kernels. We achieve this through three complementary techniques: adaptive memory hierarchy utilization based on data characteristics, optimized access patterns for coalesced memory transactions, and strategic use of shared memory for reduction operations. The kernel dynamically selects between sparse and dense processing paths, each with tailored memory access patterns to match the underlying data distribution.

\paragraph{Memory Hierarchy Utilization.}
The kernel employs a hierarchical memory strategy that adapts to the characteristics of sparse versus dense data blocks. Table~\ref{tab:memory-usage} summarizes the memory usage patterns for each processing path. Both paths leverage \texttt{\_\_ldg()} intrinsics to access read-only data through the L2 cache, bypassing L1 to reduce cache pressure. The sparse path minimizes shared memory usage (primarily for reduction operations) while maintaining 8-10 registers per thread for basic accumulation. In contrast, the dense path allocates 512B-2KB of shared memory for intermediate results and consumes 16-20 registers per thread to maintain four independent accumulators necessary for 4-way instruction-level parallelism (ILP). This adaptive approach ensures efficient resource utilization across varying data densities.
\begin{table}[h]
\centering
\caption{Memory Usage Patterns in DualGCN Kernel}
\label{tab:memory-usage}
\small
\begin{tabular}{l|c|c|l}
\hline
\textbf{Memory} & \textbf{Sparse} & \textbf{Dense} & \textbf{Purpose} \\
\hline
Global (R) & Coalesced & Strided-coal. & Input matrices \\
Global (W) & Atomic & Atomic (red.) & Output accum. \\
Shared & Minimal & 512B-2KB & Reduction \\
L2 Cache & \texttt{\_\_ldg()} & \texttt{\_\_ldg()} & Read-only \\
Registers & 8-10/thr & 16-20/thr & Accumulators \\
\hline
\end{tabular}
\end{table}

\paragraph{Coalesced Memory Access Pattern.}
To maximize memory bandwidth, we implement an access pattern ensuring coalesced transactions within each warp. Each thread $i$ accesses indices $\{i, i+N_t, i+2N_t, i+3N_t\}$, where $N_t$ is the thread count per block. This design serves three purposes:
\begin{enumerate}
    \item \textbf{Coalescing}: Consecutive threads access consecutive memory locations, achieving perfect coalescing.
    \item \textbf{ILP Optimization}: The $N_t$ stride enables independent memory requests, allowing concurrent servicing while arithmetic operations proceed on fetched data.
    \item \textbf{Bank Conflict Minimization}: Regular patterns reduce shared memory bank conflicts during reductions.
\end{enumerate}

\paragraph{Adaptive Shared Memory Configuration.}
Shared memory layout adjusts based on block dimensions. We allocate a \texttt{[WPB]} reduction buffer for partial results, with padding added for power-of-two sizes to prevent bank conflicts. The sparse path uses consecutive addresses for irregular data patterns. The dense path extends this across iterations, with threads advancing by \texttt{blockDim} between iterations (e.g., accessing \texttt{[base, base+31]}, then \texttt{[base+blockDim, base+blockDim+31]}). This maintains coalesced access for larger memory regions, ensuring optimal bandwidth utilization across diverse data distributions.
\section{Evaluation}
\begin{table}[t]
	\caption{Graph Dataset Statistics}
	\label{tab:graph-stats}
	\centering
	\scriptsize
	\setlength{\tabcolsep}{2.5pt}
	\renewcommand{\arraystretch}{1.05}
	\begin{tabular}{@{}lrr@{\hspace{5pt}}lrr@{}}
		\hline
		\textbf{Graph} & \textbf{Nodes} & \textbf{Edges}
		& \textbf{Graph} & \textbf{Nodes} & \textbf{Edges} \\
		\hline
		ARTIST
		& 50,515
		& 1,638,396
		& AMAZON0601
		& 403,394
		& 5,478,357 \\
		
		ARXIV
		& 169,343
		& 1,166,243
		& PPA
		& 576,289
		& 42,463,862 \\
		
		COLLAB
		& 235,868
		& 2,358,104
		& TWITTER
		& 580,768
		& 1,435,116 \\
		
		COM-AMAZON
		& 334,863
		& 1,851,744
		& YELP
		& 716,847
		& 13,954,819 \\
		
		YEAST
		& 1,710,902
		& 3,636,546
		& OVCAR-8H
		& 1,889,542
		& 3,946,402 \\
		
		PRODUCTS
		& 2,449,029
		& 123,718,280
		& CITATION
		& 2,927,963
		& 30,387,995 \\
		\hline
	\end{tabular}
\end{table}
\label{sec:eva}
\begin{figure*}[hbt!]
	\centering
	\includegraphics[width=0.96\textwidth,height=0.25\textwidth]{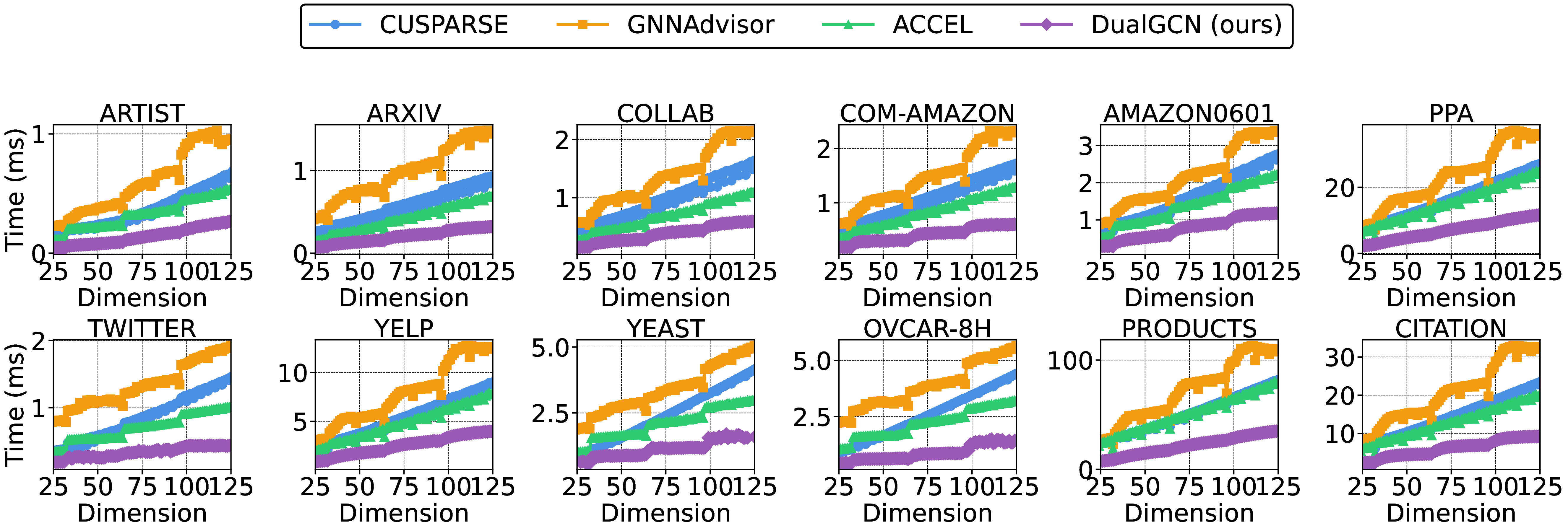}\caption{The execution times of the four GCN accelerators with different feature dimensions of the weight matrix.}\label{fig_compare_performance}
\end{figure*}
\begin{figure*}[hbt!]
	\centering
	\includegraphics[width=0.96\textwidth,height=0.25\textwidth]{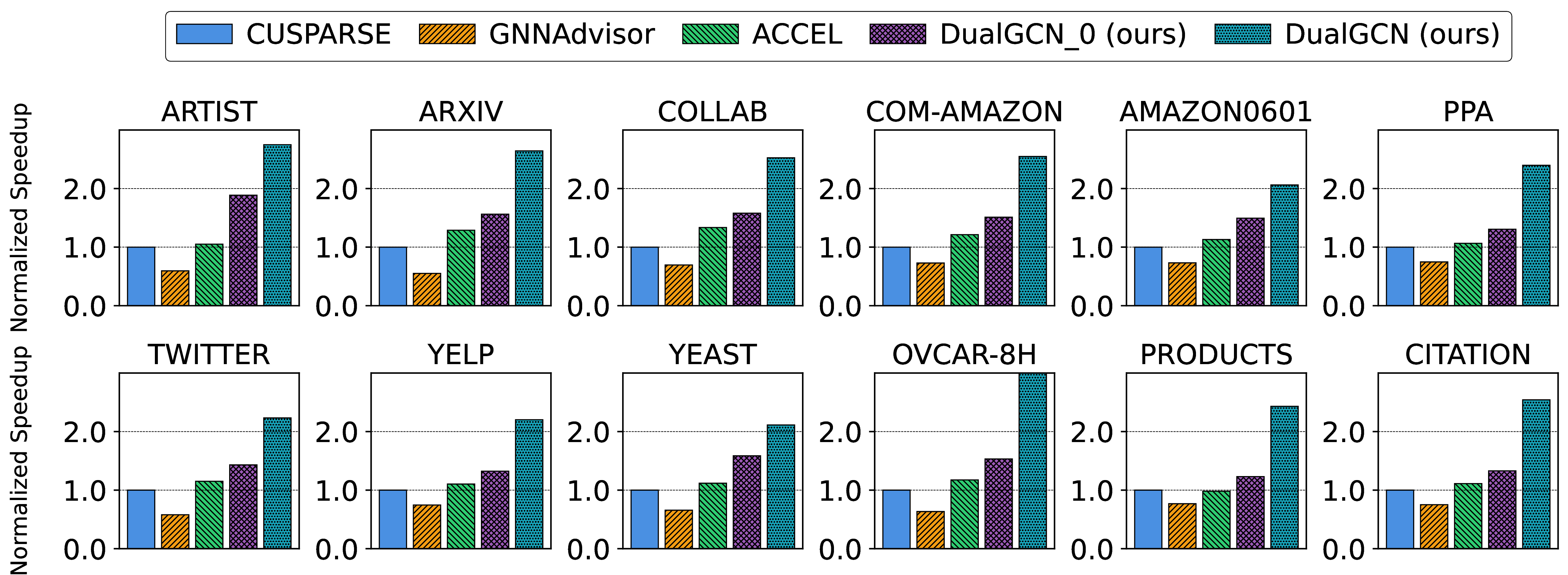}\caption{The average normalized speedup of DualGCN (ours).}\label{fig_compare_bars}
	\vspace{-2.5em}
\end{figure*}
\subsection{Experimental Environment}
Our experiments were conducted on a system running Ubuntu 20.04 equipped with an NVIDIA GeForce RTX 3090 GPU, using CUDA code compiled with NVCC 12.0. We evaluate SpMM performance using $12$ benchmark graph datasets listed in Table \ref{tab:graph-stats}. The SpMM computations, essential for GCN models, involve multiplying sparse matrices derived from these graphs with dense weight matrices having dimensions ranging from $16$ to $128$ columns. These widely-used benchmark datasets were selected for their diverse graph connectivity patterns, enabling comprehensive evaluation of our SpMM kernel optimizations for GCN workloads.

We evaluate DualGCN against leading GCN accelerators: NVIDIA's cuSPARSE
v12.0~\cite{naumov2010cusparse}, GNNAdvisor~\cite{wang2021gnnadvisor},
and ACCEL~\cite{xie2023accel}. While baselines use default partitioning,
DualGCN employs hybrid partitioning based on anonymous random walks.
Specifically, we perform $50$ walks of $30$ steps each to estimate
neighbor density and compute a hybrid metric for partitioning graph data
into sparse and dense components. Most algorithms except GNNAdvisor
achieve linear partitioning complexity and support preprocessing. We
measure only kernel execution time, as data transfer and preprocessing
overheads are comparable across implementations and executed on CPUs.
\subsection{Performance Evaluation}
Figure \ref{fig_compare_performance} demonstrates that our proposed kernel DualGCN achieves superior performance with faster execution times across 12 datasets with column dimensions ranging from $16$ to $128$. Specifically, DualGCN kernel delivers average speedups of $2.53\times$, $3.8\times$, and $2.13\times$ over cuSPARSE, GNNAdvisor, and ACCEL, respectively. These performance gains are attributed to DualGCN's efficient data partitioning scheme and optimized workload distribution, which maximize GPU platform utilization. This analysis confirms that our dual-pronged approach---dividing graphs into sparse and dense components while employing strategic resource assignment---effectively accelerates GCN-based applications.

To isolate the effect of our workload distribution strategy from data
partitioning, we conducted a controlled experiment. We constructed
DualGCN\_0, which adopts ACCEL's baseline partitioning technique while
retaining DualGCN's workload distribution method. Figure~\ref{fig_compare_bars}
compares five kernel implementations. Under the same partitioning scheme
as ACCEL, DualGCN\_0 achieves at least 52\% improvement and up to
$4.5\times$ speedup over existing methods. Combined with our proposed
partitioning approach, the complete DualGCN achieves speedups of
$1.48\times$ to $6.4\times$ across all $12$ datasets, regardless of
feature dimensions. These results demonstrate that both components
independently and jointly contribute to DualGCN's acceleration across
diverse graph structures and configurations.
\section{Conclusion}
We introduce DualGCN, a novel algorithm-architecture optimization that accelerates Graph Convolutional Networks (GCNs) on GPUs. DualGCN leverages anonymous random walks to partition graph data into sparse and dense components. We implement a dual-path workload distribution scheme that enhances data locality and parallel execution in GCN models. Our kernel orchestrates workload allocation and memory access patterns across GPU blocks and warps, leveraging the partitioned structure. Experiments demonstrate DualGCN's performance, achieving average speedups of $2.53\times$, $3.8\times$, and $2.13\times$ over cuSPARSE, GNNAdvisor, and ACCEL, respectively.
\section*{Acknowledgments}
This work was supported in part by the National Natural Science
Foundation of China (NSFC) under Grant No.~62302043, the
Interdisciplinary Project of Capital Normal University under Grant
No.~2026JCZD02, and JSPS KAKENHI under Grant No.~26K21227.
\label{sec:con}

\bibliographystyle{splncs04}
\bibliography{Graph_ISO}
\end{document}